\documentclass{llncs}
\usepackage[pdflatex=true,recompilepics=true]{gastex} 
\usepackage{gastex}

\usepackage{amsmath}
\usepackage{amssymb}
\usepackage{color}
\usepackage{epstopdf}
\usepackage{times}
\usepackage[small,compact]{titlesec} 
\usepackage[ruled]{algorithm2e}

\def\A{\mathcal{A}}
\def\B{\mathcal{B}}
\def\phi{\varphi}
\def\dist{\mathit{Dist}}
\def\dirac#1{\delta_{#1}}
\def\epsilon{\varepsilon}

\newcommand*{\dotcup}{\ensuremath{\mathaccent\cdot\cup}}

\newcommand{\TRANA}[3]{#1\xrightarrow[]{#2}#3}

\newcommand{\TRANPA}[3]{#1\xrightarrow{#2}_{{\sf P}}#3}

\newcommand{\BSP}{\ensuremath{\sim_{{\sf P}}}}

\newcommand{\ABS}[1]{|#1|}

\def\<{\langle}
\def\>{\rangle}
\def\l{\mathcal{L}}
\def\k{\mathcal{K}}

\newcommand{\ysim}[1]{\stackrel{#1}\sim}
\def\z{\mathbf{0}}
\newcommand{\TRANDA}[3]{#1\xrightarrow{#2}_{{\sf D}}#3}

\def\supp{\mathit{supp}}

\begin{document}

\title{When Equivalence and Bisimulation Join Forces in Probabilistic
  Automata\thanks{Supported by the National Natural Science Foundation
    of China (NSFC) under grant No.\ 61361136002, and Australian Research Council (ARC)
under grant Nos. DP130102764 and FT100100218. Y. F. is also supported by the Overseas Team Program of Academy of Mathematics and Systems
Science, Chinese Academy of Sciences.}} \author{Yuan
  Feng$^{1,2}$ \and Lijun Zhang$^3$} \institute{ University of Technology
  Sydney, Australia \and
  Department of Computer Science and Technology, Tsinghua University, China\\
  \and State Key Laboratory of Computer Science, Institute of
  Software, Chinese Academy of Sciences }
\maketitle

\begin{abstract}
  Probabilistic automata were introduced by Rabin in 1963 as language
  acceptors. Two automata are equivalent if and only if they accept
  each word with the same probability.  On the other side, in the
  process algebra community, probabilistic automata were re-proposed
  by Segala in 1995 which are more general than Rabin's
  automata. Bisimulations have been proposed for Segala's automata to
  characterize the equivalence between them. So far the two notions of
  equivalences and their characteristics have been studied most
  independently.  In this paper, we consider Segala's automata, and
  propose a novel notion of distribution-based bisimulation by joining the existing
  equivalence and bisimilarities. Our bisimulation bridges the two
  closely related concepts in the community, and provides a uniform way
  of studying their characteristics. We demonstrate the utility of our
  definition by studying distribution-based bisimulation metrics,
  which gives rise to a robust notion of equivalence for Rabin's
  automata.
\end{abstract}

\section{Introduction}
In 1963, Rabin \cite{Rabin63} introduced the model \emph{probabilistic
  automata} as language acceptors. In a probabilistic automaton, each
input symbol determines a stochastic transition matrix over the state
space. Starting with the initial distribution, each word (a sequence
of symbols) has a corresponding probability of reaching one of the
final states, which is referred to the accepting probability. Two
automata are equivalent if and only if they accept each word with the
same probability. The corresponding decision algorithm has been
extensively studied, see~\cite{Rabin63,Tzeng92,KieferMOWW11,KieferMOWW12}.

Markov decision processes (MDPs) were known as early as the
1950s~\cite{Bellman57}, and are a popular modeling formalism used for
instance in operations research, automated planning, and decision
support systems. In MDPs, each state has a set of enabled actions and
each enabled action leads to a distribution over successor
states. MDPs have been widely used in the formal verification of
randomized concurrent systems, and are now supported by probabilistic
model checking tools such as PRISM \cite{KwiatkowskaNP11},
MRMC~\cite{KatoenZHHJ11} and IscasMC~\cite{HLSTZ14}.

On the other side, in the context of concurrent systems, probabilistic
automata were re-proposed by Segala in 1995~\cite{Segala-thesis},
which extend MDPs with internal nondeterministic choices.  Segala's
automata are more general than Rabin's automata, in the sense that
each input symbol corresponds to one, or more than one, stochastic transition
matrices. Various behavioral equivalences are defined, including
strong bisimulations, strong probabilistic bisimulations, and weak
bisimulation extensions~\cite{Segala-thesis}. These behavioral
equivalences are used as powerful tools for state space reduction and
hierarchical verification of complex systems. Thus, their decision
algorithms~\cite{CattaniS02,BaierEM00,HermannsT12} and logical
characterizations~\cite{ParmaS07,DesharnaisGJP10,HermannsPSWZ11} are
widely studied in the literature.

Equivalences are defined for the specific initial distributions over
Rabin's automata, whereas bisimulations are usually defined over
states. For Segala's automata, state-based bisimulations have arguably
too strong distinguishing power, thus in the recent literature,
various relaxations have been proposed.  The earliest such formulation
is a distribution-based bisimulation in~\cite{DoyenHR08}, which is
defined for Rabin's automata. This is essentially an equivalent
characterization of the equivalence in the coinductive
manner, as for bisimulations. Recently, in~\cite{EisentrautHZ10}, a
distribution-based weak bisimulation has been proposed, and the
induced distribution-based strong bisimulation is further studied
in~\cite{Hennessy12}. It is shown that the distribution-based strong
bisimulation agrees with the state-based bisimulations when lifted to
distributions.

To the best of the authors' knowledge, even the two notions
are closely related, so far their characteristics have been studied
independently. As the main contribution of this paper, we consider
Segala's probabilistic automata, and propose a novel notion of
distribution-based bisimulation by joining the existing equivalence
and bisimilarities. We show that for Rabin's probabilistic automata it
coincides with equivalences, and for Segala's probabilistic automata,
it is reasonably weaker than the existing bisimulation relation.
Thus, our bisimulations bridge the two closely related concepts in the
community, and provide a uniform way of studying their
characteristics.

We demonstrate the utility of our approach by studying distribution-based bisimulation metrics. Bisimulations for probabilistic systems
are known to be very sensitive to the transition probabilities: even a
tiny perturbation of the transition probabilities will destroy
bisimilarity. Thus, bisimulation metrics have been proposed~\cite{GiacaloneJS90}: the
distance between any two states are measured, and the smaller the
distance is, the more similar they are. If the distance is zero, one then has
 the classical bisimulation.  Because of the nice property of
robustness, bisimulation metrics have attracted a lot attentions on MDPs and
their extension with continuous state space, see
\cite{DesharnaisGJP99,AlfaroMRS07,DesharnaisGJP04,DesharnaisLT08,FernsPP11,ChenBW12,Fu12,BacciBLM13,ComaniciPP12}.

All of the existing bisimulation metrics mentioned above are state-based. On the other side, as states lead to distributions in MDPs,
the metrics must be lifted to distributions. In the second part of the
paper, we propose a distribution-based bisimulation metric; we consider
it being more natural as no lifting of distances is needed. We
provide a coinductive definition as well as a fixed point
characterization, both of which are used in defining the state-based
bisimulation metrics in the literature. We provide a logical
characterization for this metric as well, and discuss the relation
of our definition and the state-based ones.

A direct byproduct of our bisimulation-based metrics is the notion of
equivalence metrics for Rabin's probabilistic automata. As for
bisimulation metrics, the equivalence metric provides a robust
solution for comparing Rabin's automata. To the best of our knowledge,
this has not been studied in the literature.  We anticipate that more
solution techniques developed in one area can inspire solutions for
the corresponding problems in the other. 

\emph{Organization of the paper.}  We introduce some notations in
Section \ref{sec:pre}. Section \ref{sec:pa} recalls the definitions of
probabilistic automata, equivalence, and bisimulation relations. We
present our distribution-based bisimulation in Section
\ref{sec:novel}, and bisimulation metrics and their logical
characterizations in \ref{sec:metric}. Section \ref{sec:conclusion}
concludes the paper. 

\section{Preliminaries}\label{sec:pre}
\paragraph{Distributions.}
For a finite set $S$, a distribution is a function $\mu:S\to
[0,1]$ satisfying $\ABS{\mu}:=\sum_{s\in S}\mu(s)= 1$. We denote by
$\mathit{Dist}(S)$ the set of distributions over $S$. We shall use
$s,r,t,\ldots$ and $\mu,\nu\ldots$ to range over $S$ and
$\mathit{Dist}(S)$, respectively. Given a set of distributions
$\{\mu_i\}_{1\leq i\leq n}$, and a set of positive weights
$\{p_i\}_{1\leq i\leq n}$ such that $\sum_{1\leq i\leq n}p_i=1$, the
\emph{convex combination} $\mu=\sum_{1\leq i\leq n}p_i\cdot\mu_i$ is
the distribution such that $\mu(s)=\sum_{1\leq i\leq
  n}p_i\cdot\mu_i(s)$ for each $s\in S$. The support of $\mu$ is
defined by $\mathit{supp}(\mu):=\{s\in S \mid \mu(s)>0\}$. For an
equivalence relation $R$ defined on $S$, we write $\mu R\nu$ if it holds
that $\mu(C)=\nu(C)$ for all equivalence classes $C\in S/R$. A
distribution $\mu$ is called \emph{Dirac} if $|\mathit{supp}(\mu)|=1$,
and we let $\dirac{s}$ denote the Dirac distribution with
$\dirac{s}(s)=1$.

Note that when $S$ is finite, the distributions
$\dist(S)$ over $S$, when regarded as a subset of $\mathbb{R}^{|S|}$,
is both convex and compact. In this paper, when we talk about convergence of distributions,
or continuity of relations such as transitions, bisimulations, and pseudometrics 
between distributions, we are referring to the normal topology of $\mathbb{R}^{|S|}$.
For a set $F\subseteq S$, we define the
(column) characteristic vector $\eta_F$ by letting $\eta_F(s)=1$ if $s\in
F$, and 0 otherwise.

\paragraph{Pseudometric.}
A pseudometric over $\dist(S)$ is a function $d : \dist(S)\times \dist(S) \rightarrow [0,1]$ such that 
(i) $d(\mu, \mu) = 0$;
(ii) $d(\mu, \nu) = d(\nu, \mu)$;
(iii) $d(\mu, \nu) + d(\nu, \omega) \geq d(\mu, \omega)$.
In this paper, we assume that a pseudometric is continuous.

\section{Probabilistic Automata and Bisimulations}\label{sec:pa}
\subsection{Probabilistic Automata}
Let $AP$ be a finite set of atomic propositions. We recall the notion of 
probabilistic automata introduced by Segala~\cite{Segala-thesis}.
\begin{definition}[Probabilistic Automata]\label{def:automata}
  A \emph{probabilistic automaton} is a tuple
  $\A=(S, Act, \rightarrow,L,\alpha)$ where $S$ is a finite set of
  states, $Act$ is a finite set of actions, ${\rightarrow}\subseteq S\times Act\times \dist(S)$ is a
  transition relation, $L:S\to 2^{AP}$ is a labeling function, and $\alpha\in \dist(S)$ is an initial distribution.
\end{definition}

As usual we only consider image-finite probabilistic automata, i.e.
 for all $s\in S$, the set $\{\mu\mid (s,a,\mu)\in{\rightarrow}\}$ is finite. A
transition $(s,a,\mu)\in{\rightarrow}$ is denoted by
$\TRANA{s}{a}{\mu}$.  We denote by $Act(s):=\{a\mid
\TRANA{s}{a}{\mu}\}$ the set of enabled actions in $s$. We say $\A$ is
\emph{input enabled}, if $Act(s)=Act$ for all $s\in S$. We say $\A$ is
an MDP if $Act$ is a singleton.

Interestingly, a subclass of probabilistic automata were already
introduced by Rabin in 1963~\cite{Rabin63}; Rabin's probabilistic automata
were referred to as \emph{reactive automata} in~\cite{Segala-thesis}. We adopt this
convention in this paper.
\begin{definition}[Reactive Automata]
  We say $\A$ is \emph{reactive} if it is input enabled, and for all $s$, $L(s)\in
  \{\emptyset,AP\}$, and $\TRANA{s}{a}{\mu} \wedge
  \TRANA{s}{a}{\mu'} \Rightarrow \mu=\mu'$.
\end{definition}

Here the condition $L(s)\in \{\emptyset,AP\}$ implies that the states
can be partitioned into two equivalence classes according to their
labeling. Below we shall identify $F:=\{s\mid L(s)=AP\}$ as the set
of \emph{accepting states}, a terminology used in reactive automata.
In a reactive automaton, each action $a\in Act$ is enabled
precisely once for all $s\in S$, thus inducing a
stochastic matrix $M(a)$ satisfying $\TRANA{s}{a}{M(a)(s,\cdot)}$.

\subsection{Probabilistic Bisimulation and Equivalence}
First, we recall the definition of (strong) probabilistic
bisimulation for probabilistic automata~\cite{Segala-thesis}.
Let $\{\TRANA{s}{a}{\mu_i}\}_{i\in I}$ be a collection of transitions, and let $\{p_i\}_{i\in I}$ be a collection of probabilities
with $\sum_{i\in I}p_i=1$. Then $(s,a,\sum_{i\in I}p_i\cdot\mu_i)$ is
called a \emph{combined transition} and is denoted by
$\TRANPA{s}{a}{\mu}$ where $\mu=\sum_{i\in I}p_i\cdot\mu_i$.  

\begin{definition}[Probabilistic bisimulation~\cite{Segala-thesis}]\label{def:bis}
  An equivalence relation $R\subseteq S\times S$ is a 
  probabilistic bisimulation if $sR r$ implies that $L(s)=L(r)$, and for each
  $\TRANA{s}{a}{\mu}$, there exists a combined transition
  $\TRANPA{r}{a}{\nu}$ such that $\mu R \nu$.

We write $s\BSP r$ whenever there is a 
probabilistic bisimulation $R$ such that $sRr$.
\end{definition}

Recently, in~\cite{EisentrautHZ10}, a distribution-based weak
bisimulation has been proposed, and the induced distribution-based
strong bisimulation is further studied in~\cite{Hennessy12}. Their
bisimilarity is shown to be the same as $\BSP$ when lifted to
distributions.  Below we recall the definition of equivalence for
reactive automata introduced by Rabin~\cite{Rabin63}.

\begin{definition}[Equivalence for Reactive Automata~\cite{Rabin63}]
  Let $\A_i=(S_i,Act_i,\rightarrow_i,L_i,\alpha_i)$ with $i=1,2$ be
  two reactive automata with $Act_1=Act_2=:Act$, and
  $F_i=\{s\in S_i\mid L(s)=AP\}$ the set of final states for $\A_i$. We say
  $\A_1$ and $\A_2$ are equivalent if $\A_1(w)=\A_2(w)$
  for each $w\in Act^*$, where $\A_i(w):=\alpha_i M_i(a_1)\ldots M_i(a_k)
  \eta_{F_i}$ provided $w=a_1\ldots a_k$.
\end{definition}

Stated in plain english, $\A_1$ and $\A_2$ with the same set of
actions are equivalent iff for an arbitrary input $w$, 
the probabilities of \emph{absorbing} in $F_1$ and $F_2$ are the same. 

So far bisimulations and equivalences were studied most
independently.  The only exception we are aware is~\cite{DoyenHR08}, in which for
Rabin's probabilistic automata, a distribution-based bisimulation is
defined that generalizes both equivalence and bisimulations. 
\begin{definition}[Bisimulation for Reactive
  Automata~\cite{DoyenHR08}] \label{def:distribution_based} Let
  $\A_i=(S_i,Act_i,\rightarrow_i,L_i,\alpha_i)$ with $i=1,2$ be two
  given reactive automata with $Act_1=Act_2=:Act$, and $F_i$ the set
  of final states for $\A_i$. A relation $R\subseteq
  \mathit{Dist}(S_1) \times \mathit{Dist}(S_2)$ is a bisimulation if
  for each $\mu R\nu$ it holds (i) $\mu\cdot\eta_{F_1} = \nu \cdot
  \eta_{F_2}$, and (ii) $(\mu M_1(a)) R (\nu M_2(a))$ for all $a\in
  Act$.

We write $\mu\sim_d\nu$ whenever there is a 
bisimulation $R$ such that $\mu R\nu$.
\end{definition}

It is shown in~\cite{DoyenHR08} that two reactive automata are equivalent if
and only if their initial distributions are distribution-based
bisimilar according to the definition above.

\section{A Novel Bisimulation Relation}\label{sec:novel}
In this section we introduce a notion of distribution-based bisimulation
for Segala's automata by extending the bisimulation defined
in~\cite{DoyenHR08}. We shall show
the compatibility of our definition with previous ones in Subsection~\ref{sec:compatibility}, and
some properties of our bisimulation in Subsection~\ref{sec:property}.

For the first step of defining a distribution-based bisimulation, we need to extend the transitions 
starting from states to those starting from distributions. A natural candidate for such an extension is 
as follows: for a distribution $\mu$ to perform an action $a$, each
state in its support \emph{must} make a combined $a$-move. However,
this definition is problematic, as in Segala's general probabilistic automata,
action $a$ may not always be enabled in any support state of $\mu$.
In this paper, we deal with this problem by 
first defining the distribution-based bisimulation (resp. distances) for input enabled automata, for which
the transition between distributions can be naturally defined, and then
reducing the equivalence (resp. distances) of two distributions
in a general probabilistic automata to the bisimilarity (resp. distances) of these distributions
in an input enabled automata which is obtained from the original one by adding a \emph{dead} state.

To make our idea more rigorous, we need some notations. For
$A\subseteq AP$ and a distribution $\mu$, we define
$\mu(A):=\sum \{ \mu(s) \mid L(s) =A \}$, which is the probability of being in those state $s$ with label $A$.

 \begin{definition}\label{def:dptran}
   We write $\TRANA{\mu}{a}{\mu'}$ if for each
   $s\in \supp(\mu)$ there exists
   $\TRANPA{s}{a}{\mu_s}$ such that 
   $\mu'=\sum_s\mu(s)\cdot\mu_s$.
\end{definition}

We first present our distribution-based bisimulation for input enabled probabilistic automata.

\begin{definition}
  Let $\A=(S, Act, \rightarrow,L,\alpha)$ be an input enabled probabilistic automaton.  A symmetric
  relation $R \subseteq \dist(S)\times \dist(S)$ is a (distribution-based) bisimulation if $\mu R\nu$ implies that
  \begin{enumerate}
  \item $\mu(A)=\nu(A)$ for each $A\subseteq AP$, and
  \item for each $a\in Act$, whenever $\TRANA{\mu}{a}{\mu'}$ then there exists a transition
  $\TRANA{\nu}{a}{\nu'}$ such that $\mu' R\nu'$.
  \end{enumerate}
We write $\mu\sim^\A\nu$ if there is a 
  bisimulation $R$ such that $\mu R \nu$.
\end{definition}

Obviously, the bisimilarity $\sim^\A$ is the largest bisimulation relation. 

For probabilistic automata which are not input enabled, we define distribution-based bisimulation with the help of
\emph{input enabled extension} specified as follows.

\begin{definition}
Let $\A=(S, Act, \rightarrow,L,\alpha)$ be a probabilistic automaton over $AP$.
The \emph{input enabled extension} of $\A$, denoted by $\A_\bot$, is defined as
an (input enabled) probabilistic automaton $(S_\bot, Act, \rightarrow^\bot, L_\bot, \alpha)$ over $AP_\bot$ where
\begin{enumerate}
\item  $S_\bot = S\cup \{\bot\}$ where $\bot$ is a \emph{dead} state
not in $S$;
\item $AP_\bot=AP\cup \{dead\}$ with $dead\not\in AP$;
\item $\rightarrow^\bot {=} \rightarrow \cup~ \{(s, a, \dirac{\bot}) \mid a\not \in Act(s)\} \cup \{(\bot, a, \dirac{\bot}) \mid a \in Act\}$;
\item $L_\bot(s) = L(s)$ for any $s\in S$, and $L_\bot(\bot)=\{dead\}$.
\end{enumerate} 
\end{definition}

\begin{definition}
  Let $\A$ be a probabilistic automaton which is not input enabled. Then 
  $\mu$ and $\nu$ are bisimilar, denoted by $\mu\sim^\A\nu$, 
  if $\mu\sim^{\A_\bot}\nu$ in $\A_\bot$.
\end{definition}

We always omit the superscript $\A$ in $\sim^\A$ when no confusion arises.

\subsection{Compatibility}\label{sec:compatibility}
In this section we instantiate appropriate labeling functions and show
that our notion of bisimilarity is a conservative extension of both
probabilistic bisimulation~\cite{Rabin63} 
and equivalence relations~\cite{DoyenHR08}.

\begin{lemma}\label{lem:bsp}
 Let $\A$ be a probabilistic automaton where $AP=Act$, and $L(s) =Act(s)$ for each $s$. Then, $\mu\BSP \nu$ implies
 $\mu \sim \nu$.
\end{lemma}
\begin{proof}
First, it is easy to see that for a given probabilistic automata $\A$ with $AP=Act$ and $L(s) =Act(s)$ for each $s$, and distributions $\mu$ and $\nu$ in $\dist(S)$, 
$\mu \BSP\nu$ in $\A$ if and only if $\mu \BSP\nu$ in the input enabled extension $\A_\bot$. Thus 
we can assume without loss of any generality that $\A$ itself is input enabled. 

It suffices to show that the symmetric relation
$$R = \{ (\mu, \nu) \mid \mu \BSP\nu \}$$
is a bisimulation. For each $A\subseteq Act$, let $S(A) = \{s\in S \mid  L(s)= A\}$. Then $S(A)$ is the disjoint union of some equivalence classes of $\BSP$; that is, $S(A) = \dotcup\{M\in S/{\BSP} \mid M\cap S(A) \neq \emptyset\}$. Suppose $\mu\BSP\nu$. Then for any $M\in S/{\BSP}$, $\mu(M) = \nu(M)$, hence $\mu(A) = \mu(S(A)) = \nu(S(A)) = \nu(A).$ 

Let $\TRANA{\mu}{a}{\mu'}$. Then for any $s\in S$ there exists $\TRANPA{s}{a}{\mu_s}$ such that
$$\mu' = \sum_{s\in S} \mu(s)\cdot \mu_s.$$
Now for each $t\in S$, let $[t]_{\BSP}$ be the equivalence class of $\BSP$ which contains $t$. Then for every $s \in [t]_{\BSP}$, to match the transition $\TRANPA{s}{a}{\mu_s}$ there exists some $\nu_t^s$ such that $\TRANPA{t}{a}{\nu^s_t}$ and $\mu_s\BSP \nu^s_t$. Let 
$$\nu_t = \sum_{s\in [t]_{\BSP}}\frac{\mu(s)}{\mu([t]_{\BSP})}\cdot \nu_t^s.$$
Then we have $\TRANPA{t}{a}{\nu_t}$, and $\TRANA{\nu}{a}{\nu'}$ where 
$$\nu' := \sum_{t\in S} \nu(t) \cdot\nu_t.$$
It remains to prove $\mu' \BSP \nu'$. For any $M\in S/{\BSP}$, since $\mu_s\BSP \nu^s_t$ we have 
\begin{eqnarray*}
\nu_t(M) &=& \sum_{s\in [t]_{\BSP}}\frac{\mu(s)}{\mu([t]_{\BSP})} \nu_t^s(M)= \sum_{s\in [t]_{\BSP}}\frac{\mu(s)}{\mu([t]_{\BSP})} \mu_s(M).
\end{eqnarray*} 
Thus
\begin{eqnarray*}
\nu'(M) &=&  \sum_{t\in S} \nu(t) \sum_{s\in [t]_{\BSP}}\frac{\mu(s)}{\mu([t]_{\BSP})} \mu_s(M)\\
&=& \sum_{s\in S}  \mu(s)\mu_s(M)\sum_{t\in [s]_{\BSP}}\frac{\nu(t)}{\nu([s]_{\BSP})}\\
&=& \sum_{s\in S} \mu(s) \mu_s(M)  = \mu'(M)
\end{eqnarray*} 
where for the second equality we have used the fact that $\mu([t]_{\BSP}) = \nu([s]_{\BSP})$ for any $s\BSP t$. \qed
\end{proof}

Probabilistic bisimulation is defined over states inside one
automaton, whereas equivalence and distribution for reactive automata are
defined over two automata. However, they can be connected by the notion of
direct sum of two automata, which is the automaton obtained by
considering the disjoint union of states, edges and labeling functions
respectively.

\begin{lemma}\label{lem:eqra}
  Let $\A_1$ and $\A_2$ be two reactive automata with the same set of
  actions $Act$. Let $F_i=\{s\in S_i \mid L(s)=AP\}$. Then, the following are equivalent:
  \begin{enumerate}
  \item $\A_1$ and
  $\A_2$ are equivalent,
\item $\alpha_1\sim_d\alpha_2$, 
\item $\alpha_1\sim \alpha_2$ in their direct sum. 
  \end{enumerate}
\end{lemma}
\begin{proof}
  The equivalence between (1) and (2) is shown in~\cite{DoyenHR08}. 
  The equivalence between (2) and (3) is straightforward, as for
  reactive automata our definition degenerates to Definition \ref{def:distribution_based}. \qed
\end{proof}

To conclude this section, we present an example to show that our bisimilarity is \emph{strictly} weaker than
$\BSP$.
\begin{example}\label{exa:lmc}
  Consider the example probabilistic automaton depicted in
  Fig.~\ref{fig:exam1}, which is inspired from an example
  in~\cite{DoyenHR08}. 
 Let $AP=Act=\{a\}$, $L(s) =Act(s)$ for each $s$, and $\epsilon_1=\epsilon_2=0$. We argue that
  $q\not\BSP q'$. Otherwise, note $\TRANA{q}{a}{\frac12 \dirac{r_1} +
    \frac12 \dirac{r_2}}$ and $\TRANA{q'}{a}{\dirac{r'}}$. Then we must have $r' \BSP r_1
  \BSP r_2$. This is impossible, as $\TRANA{r_1}{a}{\frac23
    \dirac{s_1} + \frac13 \dirac{s_2}}$ and $\TRANA{r'}{a}{\frac12
    \dirac{s_1'} + \frac12 \dirac{s_2'}}$, but $s_1\BSP s_1'\not\BSP s_2\BSP s_2'$.

However, by our definition of bisimulation, the Dirac distributions $\dirac{q}$ and $\dirac{q'}$ are indeed bisimilar. The reason is, we have the following transition
$$\TRANA{\frac12 \dirac{r_1} + \frac12\dirac{r_2}}{a}{\frac13 \dirac{s_1} + \frac16 \dirac{s_2} + \frac16 \dirac{s_3} + \frac13 \dirac{s_4}},$$
and it is easy to check $\dirac{s_1}\sim \dirac{s_3}\sim \dirac{s_1'}$ and $\dirac{s_2}\sim \dirac{s_4}\sim \dirac{s_2'}$. Thus we have $\frac12 \dirac{r_1} + \frac12 \dirac{r_2} \sim \dirac{r'}$, and finally $\dirac{q}\sim \dirac{q'}$.
\begin{figure}[tbh]
\begin{center} \scalebox{0.85}{
    \begin{gpicture}[name=loop](40,28)(0,0)
  \node(q)(20,20){$q$}
  \node(r1)(0,0){$r_1$}
  \node(r2)(40,0){$r_2$}
  \drawedge[ELside=r](q,r1){$a$, $\frac12$}
  \drawedge(q,r2){$a$, $\frac12$}
  \node(s1)(-10,-20){$s_1$}
  \node(s2)(7.5,-20){$s_2$}
  \node(s3)(32.5,-20){$s_3$}
  \node(s4)(50,-20){$s_4$}
  \drawedge[ELside=r](r1,s1){$a$, $\frac23+\epsilon_1$}
  \drawloop[loopdiam=5,loopangle=180](s1){$a$}
  \drawedge(r1,s2){$a$, $\frac13-\epsilon_1$}
  \drawloop[loopdiam=5,loopangle=180](s3){$a$}
  \drawedge[ELside=r](r2,s3){$a$, $\frac13-\epsilon_2$}
  \drawedge(r2,s4){$a$, $\frac23+\epsilon_2$}
\end{gpicture} 
\hspace{2cm}
   \begin{gpicture}[name=loop](40,28)(0,0)
  \node(q)(20,20){$q'$}
  \node(r)(20,0){$r'$}
  \drawedge[ELside=r](q,r){$a$, $1$}
  \node(s1)(10,-20){$s_1'$}
  \drawloop[loopdiam=5,loopangle=180](s1){$a$}
  \node(s2)(30,-20){$s_2'$}
  \drawedge[ELside=r](r,s1){$a$, $\frac12$}
  \drawedge(r,s2){$a$, $\frac12$}
\end{gpicture}}
\end{center}
\caption{An illustrating example in which state labelings are defined by $L(s)=Act(s)$.}\label{fig:exam1}
\end{figure}
\end{example}

\subsection{Properties of the Relations}\label{sec:property}
In the following, we show that the notion of bisimilarity is in harmony with the linear combination and the limit of distributions. 

\begin{definition}\label{def:decom}
 A binary relation $R \subseteq
  \dist(S) \times\dist(S)$
 is said to be 
 \begin{itemize}
 \item \emph{linear}, if for any finite set $I$ and any probabilistic distribution $\{p_i\}_{i\in I}$, $\mu_i R \nu_i$ for each $i$
 implies
 $(\sum_{i\in I}p_i\cdot\mu_i)R(\sum_{i\in I}p_i\cdot\nu_i)$;
 \item \emph{continuous}, if for any convergent sequences of distributions $\{\mu_i\}_i$ and $\{\nu_i\}_i$, $\mu_i R \nu_i$ for each $i$ implies $(\lim_i \mu_i) R (\lim_i\nu_i)$;
 \item \emph{left-decomposable}, if
 $(\sum_{i\in I}p_i\cdot\mu_i)R\nu$, where $0< p_i \leq 1$ and $\sum_{i\in I}p_i = 1$, then $\nu$ can be written as $\sum_{i\in I}p_i\cdot\nu_i$
    such that $\mu_i R\nu_i$ for every $i \in I$.
\item \emph{left-convergent}, if
 $(\lim_i \mu_i)R\nu$, then for any $i$ we have $\mu_i R\nu_i$ for some $\nu_i$ with $\lim_i \nu_i = \nu$.
\end{itemize}
\end{definition}

We prove below that our transition relation between distributions satisfies these properties.
 
\begin{lemma}\label{lem:tranld} For an input enabled probabilistic automata, the transition relation $\TRANA{}{a}{}$ between distributions is linear, continuous, left-decomposable,
and left-convergent.
\end{lemma}
\begin{proof}
 \begin{itemize}
 \item \emph{Linearity}. Let $I$ be a finite index set and $\{p_i\mid i\in I\}$ a probabilistic distribution on $I$. 
Suppose $\TRANA{\mu_i}{a}{\nu_i}$ for each $i \in I$.  Then by definition, for each
   $s$ there exists
   $\TRANPA{s}{a}{\mu^i_s}$ such that 
   $\nu_i=\sum_{s}\mu_i(s)\cdot\mu^i_s$.
Now let $\mu=\sum_{i\in I}p_i\cdot\mu_i$. Then for each $s\in \supp(\mu)$,
$$\TRANPA{s}{a}\mu_s:=\sum_{i\in I}{\frac{p_i\mu_i(s)}{\mu(s)}\cdot \mu^i_s}.$$
 On the other hand, we check that 
\begin{eqnarray*}\nu := \sum_{i\in I} p_i \cdot \nu_i &=& 
\sum\limits_{s\in S}\sum_{i\in I} p_i\mu_i(s)\cdot\mu^i_s = \sum\limits_{s\in S }\mu(s)\cdot\mu_s .
\end{eqnarray*}
 Thus $\TRANA{\mu}{a}{\nu}$ as expected.
 
\item \emph{Continuity}. Suppose $\TRANA{\mu_i}{a}{\nu_i}$ for each $i \in I$, and $\lim_i \mu_i = \mu$.  By definition, for each
   $s$ there exists
   $\TRANPA{s}{a}{\mu^i_s}$ such that 
   $\nu_i=\sum_{s}\mu_i(s)\cdot\mu^i_s$.
Note that $\dist(S)$ is a compact set. For each $s$ we can choose a convergent subsequence $\{\mu_s^{i_k}\}_k$ of $\{\mu_s^{i}\}_i$ such that
$\lim_k \mu_s^{i_k} = \mu_s$ for some $\mu_s$. Then $\TRANPA{s}{a}{\mu_s}$, and 
\begin{eqnarray*}\TRANA{\mu}{a}{\nu} :=  \sum_{s\in S}\mu(s)\cdot\mu_s.
\end{eqnarray*}
Note that for each $k$,
 \begin{eqnarray*} \|\nu_{i_k} - \nu\|_1 &\leq & \|\mu_{i_k} - \mu\|_1 + \sum\limits_{s\in S} \mu(s) \|\mu^{i_k}_s - \mu_s\|_1
\end{eqnarray*}
where $\|\cdot\|_1$ denotes the $l_1$-norm.
We have $\nu=\lim_k\nu_{i_k}$ by the assumption that $\lim_i \mu_i = \mu$. 
Thus $\lim_i \nu_i = \nu$, as $\{\nu_i\}_i$ itself converges.

 \item \emph{Left-decomposability}. Let
 $\mu:=\TRANA{(\sum_{i\in I}p_i\cdot\mu_i)}{a}{\nu}$. Then by definition, for each
   $s$ there exists
   $\TRANPA{s}{a}{\mu_s}$ such that 
   $\nu=\sum_{s}\mu(s)\cdot\mu_s$. Thus 
$$\TRANA{\mu_i}{a}\nu_i:=\sum\limits_{s\in S }\mu_i(s)\cdot\mu_s.$$
Finally, it is easy to show that $\sum_{i\in I} p_i \cdot \nu_i = \nu$.  
 
 \item \emph{Left-convergence}. Similar to the last case.   \qed
\end{itemize}
\end{proof}
%
%
%

\begin{theorem}\label{thm:ld}
The bisimilarity relation $\sim$ is both linear and continuous.
\end{theorem}
\begin{proof}
Note that if $\mu_i\in \dist(S)$ for any $i$, then both $\sum_i p_i\cdot \mu_i$ and $\lim_i \mu_i$ (if exists) are again in $\dist(S)$.
Thus we need only consider the case when the automata is input enabled.
\begin{itemize}
\item
 \emph{Linearity}. It suffices to show that the symmetric relation 
 \begin{eqnarray*}
 R&=&\left\{\left(\sum_{i\in I} p_i\cdot \mu_i, \sum_{i\in I} p_i\cdot \nu_i\right) \mid I \mbox{
finite}, \sum_{i\in I} p_i =1, \forall i.(p_i\geq0 \wedge \mu_i \sim\nu_i)\right\}
 \end{eqnarray*}
is a bisimulation. Let $\mu=\sum_{i\in I}p_i\cdot\mu_i$, $\nu=\sum_{i\in I}p_i\cdot\nu_i$, and $\mu R\nu$. 
 Then for any $A\subseteq AP$, 
 $$\mu(A) = \sum_{i\in I} p_i\cdot \mu_i(A) = \sum_{i\in I} p_i\cdot \nu_i(A)
  = \nu(A).$$
 Now suppose $\TRANA{\mu}{a}{\mu'}$. 
 Then by Lemma~\ref{lem:tranld} (left-decomposability),  for each $i\in I$ we have
 $\TRANA{\mu_i}{a}{\mu'_i}$ for some $\mu'_i$ such that $\mu' = \sum_{i}p_i\cdot \mu'_i$. From the assumption that
 $\mu_i \sim\nu_i$, we derive $\TRANA{\nu_i}{a}{\nu'_i}$ with $\mu'_i \sim\nu'_i$ for each $i$. Thus $\TRANA{\nu}{a}{\nu'}
 :=\sum_{i}p_i\cdot \nu'_i$ by Lemma~\ref{lem:tranld} again (linearity). Finally, it is obvious that $(\mu', \nu')\in R$. 
 
 \item \emph{Continuity}.  It suffices to show that the symmetric relation
\[ R=\{ (\mu, \nu) \mid \forall  i\geq 1, \mu_i\sim{\nu_i}, \lim_i \mu_i = \mu, \mbox{ and }\lim_i \nu_i = \nu\}\]
is a bisimulation. First, for any $A\subseteq AP$, we have
\[\mu(A) = \lim_i \mu_i(A) = \lim_i \nu_i(A) = \nu(A). \]
Let $\TRANA{\mu}{a}{\mu'}$. By Lemma~\ref{lem:tranld} (left-convergence),  for any $i$ we have
$\TRANA{\mu_i}{a}{\mu_i'}$ with $\lim_{i} \mu_i' =\mu'.$
 To match the transitions, we have $\TRANA{\nu_i}{a}{\nu'_i}$ such that $\mu_i'\sim{\nu'_i}$.
Note that $\dist(S)$ is a compact set. We can choose a convergent subsequence $\{\nu_{i_k}'\}_k$ of $\{\nu_{i}'\}_i$ such that
$\lim_k \nu_{i_k}' = \nu'$ for some $\nu'$. From and fact that $\lim_i \nu_i = \nu$ and Lemma~\ref{lem:tranld} (continuity), it holds 
$\TRANA{\nu}{a}{\nu'}$ as well. Finally, it is easy to see that $(\mu', \nu')\in R$. 
 \end{itemize}
 \qed
\end{proof}

In general, our definition of bisimilarity is not
left-decomposable. This is in sharp contrast with the 
bisimulations defined by using the lifting technique~\cite{DengGHM09}. However, this
should not be regarded as a shortcoming; actually it is the key
requirement we abandon in this paper, which makes our definition
reasonably weak. This has been clearly illustrated in 
Example~\ref{exa:lmc}.

\section{Bisimulation Metrics}\label{sec:metric}
We present distribution-based bisimulation metrics with discounting factor $\gamma\in (0,1]$ in
this section.  Three different ways of defining bisimulation
metrics between states exist in the literature: one coinductive
definition based on bisimulations~\cite{YingW00,Ying01,Ying02,DesharnaisLT08}, one based on
the maximal logical differences~\cite{DesharnaisGJP99,DesharnaisGJP04,BreugelSW07},
and one on fixed point~\cite{AlfaroMRS07,BreugelSW07,FernsPP11}. We propose all the three
versions for our distribution-based bisimulations with discounting. Moreover, we show that they coincide.  We fix a discount
factor $\gamma\in (0,1]$ throughout this section. For any $\mu, \nu\in \dist(S)$, we define the distance
$$d_{AP}(\mu, \nu) := \frac 12 \sum_{A\subseteq AP} \left|\mu(A) - \nu(A)\right|.$$
Then it is easy to check that $$d_{AP}(\mu, \nu) = \max_{\B\subseteq 2^{AP}} \left|\sum_{A\in \B} \mu(A) - \sum_{A\in \B} \nu(A)\right|= \max_{\B\subseteq 2^{AP}} \left[\sum_{A\in \B} \mu(A) - \sum_{A\in \B} \nu(A)\right].$$
\subsection{A Direct Approach}

\begin{definition}
Let $\A=(S, Act, \rightarrow,L,\alpha)$ be an input enabled probabilistic automaton.  A family of symmetric relations $\{R_\epsilon \mid \epsilon \geq 0\}$ over $\dist(S)$ is a (discounted) approximate bisimulation if
for any $\epsilon\geq 0$ and $\mu R_\epsilon\nu$, we have 
  \begin{enumerate}
\item $d_{AP}(\mu, \nu)\leq \epsilon$;
\item 
for each $a\in Act$,  $\TRANA{\mu}{a}{\mu'}$ implies that there exists a transition
  $\TRANA{\nu}{a}{\nu'}$ such that $\mu' R_{\epsilon/\gamma}\nu'$.
\end{enumerate}
We write $\mu\sim_\epsilon^\A\nu$ whenever there is an approximate
 bisimulation $\{R_\epsilon \mid \epsilon\geq 0\}$ such that $\mu R_\epsilon \nu$.  For any two
distributions $\mu$ and $\nu$, we define the bisimulation distance
of $\mu$ and $\nu$ as

\begin{equation}\label{def:metric}
D_b^\A(\mu, \nu) = \inf\{\epsilon \geq 0 \mid \mu\sim^\A_{\epsilon} \nu\}.
\end{equation}
\end{definition}

Again, the approximate bisimulation and bisimulation distance of distributions in a general probabilistic automaton can be
defined in terms of the corresponding notions in the input enabled extension; that is, $\mu\sim_\epsilon^\A\nu$ if $\mu\sim_\epsilon^{\A_\bot}\nu$, and $D_b^\A(\mu, \nu) := D_b^{\A_\bot}(\mu, \nu)$. We always omit the superscripts for simplicity if no confusion arises.

It is standard to show that the family $\{\sim_\epsilon \mid \epsilon\geq 0\}$ is itself an approximate bisimulation.
The following lemma collects some more properties of $\sim_\epsilon$.

\begin{lemma}\label{lem:abis}
\begin{enumerate}
\item For each $\epsilon$, the $\epsilon$-bisimilarity $\sim_\epsilon$ is both
linear and continuous. 
\item If $\mu\sim_{\epsilon_1}\nu$ and $\nu\sim_{\epsilon_2}\omega$, then $\mu\sim_{\epsilon_1+\epsilon_2}\omega$;
\item $\sim_{\epsilon_1}{\subseteq}\sim_{\epsilon_2}$ whenever $\epsilon_{1}\leq \epsilon_{2}$.
\end{enumerate}
\end{lemma}
\begin{proof} 
The proof of item 1 is similar to
Theorem~\ref{thm:ld}.
For item 2, it suffices to show that $\{R_\epsilon \mid \epsilon\geq 0\}$ where $R_\epsilon\ {=} \bigcup_{\epsilon_1+\epsilon_2 = \epsilon}\left(\sim_{\epsilon_1}\circ\sim_{\epsilon_2}\right)$ is an approximate bisimulation (in the extended automata, if necessary), which is routine. For item 3, suppose $\epsilon_2>0$. Then it is  easy to show $\{R_\epsilon \mid \epsilon\geq 0\}$,  $R_\epsilon\ {=} \sim_{\epsilon\epsilon_1/\epsilon_2}$, is an approximate bisimulation. 
Now if  $\mu\sim_{\epsilon_1}\nu$, that is, $\mu\sim_{\epsilon_2\epsilon_1/\epsilon_2}\nu$, then $\mu R_{\epsilon_2}\nu$, and thus $\mu\sim_{\epsilon_2}\nu$ as required. \qed
\end{proof}

The following theorem states that the infimum in the definition
Eq.~\eqref{def:metric} of bisimulation distance can be replaced by
minimum; that is, the infimum is achievable.
\begin{theorem}\label{thm:metric}
For any $\mu,\nu\in \dist(S)$,  $\mu\sim_{D_b(\mu,\nu)} \nu.$
\end{theorem}
\begin{proof} By definition, we need to prove $\mu\sim_{D_b(\mu,\nu)} \nu$ in the extended automaton. We first prove that for any $\epsilon\geq0$, the symmetric relations $\{R_\epsilon \mid \epsilon \geq 0\}$ where 
\begin{eqnarray*}
R_\epsilon=\{(\mu, \nu)& \mid &\mu \sim_{\epsilon_{i}} \nu \mbox{ for each } \epsilon_{1}\geq\epsilon_{2}\geq\cdots \geq 0, \mbox{ and } \lim_{i\rightarrow \infty} \epsilon_{i}=\epsilon\}
\end{eqnarray*}
is an approximate bisimulation. Suppose $\mu R_\epsilon\nu$. Since $\mu\sim_{\epsilon_{i}} \nu$ we have $d_{AP}(\mu, \nu)\leq \epsilon_i$ for each $i$. Thus $d_{AP}(\mu, \nu)\leq \epsilon$ as well. 
Furthermore, if $\TRANA{\mu}{a}{\mu'}$, then for any $i\geq 1$,  $\TRANA{\nu}{a}{\nu_{i}}$ and
$\mu'\sim_{\epsilon_{i}/\gamma}\nu_{i}$. Since $\dist(S)$ is compact, there exists a subsequence $\{\nu_{i_k}\}_k$ of $\{\nu_i\}_i$ such that $\lim_k \nu_{i_k} = \nu'$ for some $\nu'$. We claim that
\begin{itemize}
\item $\TRANA{\nu}{a}{\nu'}$. This follows from the continuity of the transition $\TRANA{}{a}{}$, Lemma~\ref{lem:tranld}.
\item For each $k\geq 1$,  $\mu'\sim_{\epsilon_{i_k}/\gamma}\nu'$. Suppose conversely that $\mu'\not\sim_{\epsilon_{i_k}/\gamma}\nu'$ for some $k$. Then by the continuity of $\sim_{\epsilon_{i_k}/\gamma}$, we have $\mu'\not\sim_{\epsilon_{i_k}/\gamma}\nu_j$ for some $j\geq i_k$. This contradicts the fact that  $\mu'\sim_{\epsilon_j/\gamma}\nu_{j}$ and Lemma~\ref{lem:abis}(3). Thus $\mu' R_{\epsilon/\gamma}\nu'$ as required.
\end{itemize}

Finally, it is direct from definition that there exists a decreasing sequence $\{\epsilon_i\}_i$ such that $\lim_i \epsilon_i = D_b(\mu, \nu)$ and $\mu \sim_{\epsilon_{i}} \nu$ for each $i$. Then the theorem follows.
\qed
\end{proof}

A direct consequence of the above theorem is that the bisimulation distance between two distributions vanishes if and only if they are  bisimilar.
\begin{corollary} For any $\mu, \nu\in \dist(S)$, $\mu\sim \nu$ if and only if $D_b(\mu,\nu)=0$. \end{corollary}
\begin{proof} Direct from Theorem~\ref{thm:metric}, by noting that $\sim{=}\sim_{0}$. \qed\end{proof}

The next theorem shows that $D_b$ is indeed a pseudometric.
\begin{theorem}
The bisimulation distance $D_b$ is a pseudometric on $\dist(S)$.
\end{theorem}
\begin{proof} We need only to prove that $D_b$ satisfies the triangle inequality
$$D_b(\mu,\nu) + D_b(\nu,\omega) \geq D_b(\mu,\omega).$$
By Theorem~\ref{thm:metric}, we have $\mu\sim_{D_b(\mu,\nu)}\nu$ and $\nu\sim_{D_b(\nu,\omega)}\omega$. Then the result follows from Lemma~\ref{lem:abis}(2). \qed
\end{proof}

\subsection{Modal Characterization of the Bisimulation Metrics}
We now present a Hennessy-Milner type modal logic motivated by~\cite{DesharnaisGJP99,DesharnaisGJP04} to characterize the distance between distributions.

\begin{definition} The class $\l_m$ of modal formulae over $AP$, ranged over by $\phi$, $\phi_1$, $\phi_2$, etc, is defined by the following grammar:
\begin{eqnarray*}
\phi &::=& \B \ |\ \phi\oplus p\ |\ \neg \phi\ |\  \bigwedge_{i\in I} \phi_i\ |\ \<a\>\phi
\end{eqnarray*}
where $\B\subseteq 2^{AP}$, $p\in [0,1]$, $a\in Act$, and $I$ is an index set.
\end{definition}

Given an input enabled probabilistic automaton $\A=(S, Act, \rightarrow,L,\alpha)$ over $AP$,
instead of defining the satisfaction relation $\models$ for the
qualitative setting, the
(discounted) semantics of the logic $\l_m$ is given in terms of
functions from $\dist(S)$ to $[0,1]$.  For any formula $\phi\in \l_m$,
the satisfaction function of $\phi$, denoted by $\phi$ again for
simplicity, is defined in a structural induction way as follows:
\begin{itemize}
\item $\B(\mu) := \sum_{A\in \B}\mu(A)$;
\item $(\phi\oplus p)(\mu) := \min\{\phi(\mu)+p, 1\}$;
\item $(\neg \phi)(\mu) := 1-\phi(\mu)$;
\item $ (\bigwedge_{i\in I} \phi_i)(\mu) := \inf_{i\in I} \phi_{i}(\mu)$;
\item $(\<a\>\phi)(\mu) := \sup_{\TRANA{\mu}{a}{\mu'}} \gamma \cdot\phi(\mu')$.
\end{itemize}


\begin{lemma}\label{lem:continuous}
For any $\phi\in \l_m$, $\phi : \dist(S) \rightarrow [0,1]$ is a continuous function.
\end{lemma}
\begin{proof} We prove by induction on the structure of $\phi$. The basis case when $\phi \equiv \B$ is obvious. The case of $\phi \equiv \phi'\oplus p$, $\phi \equiv \neg \phi'$, and $\phi \equiv \bigwedge_{i\in I} \phi_i$ are all easy from induction.
In the following we only consider the case when $\phi \equiv \<a\>\phi'$. 

Take arbitrarily $\{\mu_i\}_i$ with $\lim_i \mu_i = \mu$. We need to show there exists a subsequence $\{\mu_{i_k}\}_k$ of $\{\mu_i\}_i$ such that $\lim_k\phi(\mu_{i_k}) = \phi(\mu)$. Take arbitrarily $\epsilon >0$. 
\begin{itemize}
\item Let $\mu^*\in \dist(S)$ such that $\TRANA{\mu}{a}{\mu^*}$ and $\phi(\mu)\leq \gamma \cdot\phi'(\mu^*)+\epsilon/2$. We have from the left-convergence of $\TRANA{}{a}{}$ that $\TRANA{\mu_i}{a}{\nu_i}$ for some $\nu_i$, and $\lim_i \nu_i = \mu^*$. By induction, $\phi'$ is a continuous function. Thus we can find $N_1\geq 1$ such that for any $i\geq N_1$, $|\phi'(\mu^*)-\phi'(\nu_i)|<\epsilon/2\gamma$.

\item For each $i\geq 1$, let $\mu_i^*\in \dist(S)$ such that $\TRANA{\mu_i}{a}{\mu_i^*}$ and $\phi(\mu_i)\leq \gamma \cdot\phi'(\mu_i^*)+\epsilon/2$. Then we have $\TRANA{\mu}{a}{\nu^*}$ with $\nu^*=\lim_k \mu_{i_k}^*$ for some convergent subsequence $\{\mu_{i_k}^*\}_k$ of $\{\mu_i^*\}_i$. Again, from the induction that $\phi'$ is continuous, we can find $N_2\geq 1$ such that for any $k\geq N_2$, $|\phi'(\mu_{i_k}^*)-\phi'(\nu^*)|<\epsilon/2\gamma $.
\end{itemize}
Let $N=\max\{N_1, N_2\}$. Then for any $k\geq N$, we have from  $\TRANA{\mu}{a}{\nu^*}$ that
\begin{eqnarray*}
\phi(\mu_{i_k}) - \phi(\mu) &\leq & \gamma [\phi'(\mu_{i_k}^*) -  \phi'(\nu^*)] + \gamma \cdot\phi'(\nu^*) - \phi(\mu) + \epsilon/2\\
&\leq& \gamma [\phi'(\mu_{i_k}^*) -  \phi'(\nu^*)]+ \epsilon/2 < \epsilon.
\end{eqnarray*}
Similarly, from $\TRANA{\mu_{i_k}}{a}{\nu_{i_k}}$ we have
\begin{eqnarray*}
\phi(\mu) - \phi(\mu_{i_k})  &\leq& \gamma [\phi'(\mu^*) - \phi'(\nu_{i_k})] + \gamma \cdot\phi'(\nu_{i_k}) - \phi(\mu_{i_k}) + \epsilon/2 \\
&\leq&\gamma [\phi'(\mu^*) - \phi'(\nu_{i_k})]+ \epsilon/2 < \epsilon.
\end{eqnarray*}
Thus $\lim_k\phi(\mu_{i_k}) = \phi(\mu)$ as required.
\qed
\end{proof}

From Lemma~\ref{lem:continuous}, and noting that the set $\{\mu' \mid
\TRANA{\mu}{a}{\mu'}\}$ is compact for each $\mu$ and $a$, the
supremum in the semantic definition of $\<a\>\phi$ can be replaced by
maximum; that is, $(\<a\>\phi)(\mu) =
\max_{\TRANA{\mu}{a}{\mu'}}\gamma \cdot \phi(\mu')$.
Now we define the logical distance for
distributions.
\begin{definition} The \emph{logic distance} of $\mu$ and $\nu$ in $\dist(S)$ of an input enabled automaton is defined by 
  \begin{align}
    \label{eq:2}
D_l^\A(\mu, \nu) = \sup_{\phi \in \l_m} |\phi(\mu) - \phi(\nu)|\ .    
  \end{align}
The logic distance for a general probabilistic automaton can be
defined in terms of the input enabled extension; that is, $D_l^\A(\mu, \nu) := D_l^{\A_\bot}(\mu, \nu)$. 
We always omit the superscripts for simplicity.
\end{definition}

Now we can show that the logic distance exactly coincides with bisimulation distance for any distributions.

\begin{theorem}\label{thm:dbdl}
$D_b = D_l$. 
\end{theorem}
\begin{proof}
As both $D_b$ and $D_l$ are defined in terms of the input enabled extension of automata, we only need to prove the result for
input enabled case.
Let $\mu, \nu\in \dist(S)$. We first prove $D_b(\mu, \nu) \geq D_l(\mu, \nu)$. It suffices to show by structural induction that for any $\phi\in \l_m$, 
$|\phi(\mu) - \phi(\nu)|\leq D_b(\mu, \nu).$
There are five cases to consider.
\begin{itemize}
\item $\phi \equiv \B$ for some $\B\subseteq 2^{AP}$. Then $|\phi(\mu) - \phi(\nu)| = |\sum_{A\in \B} [\mu(A) -\nu(A)]|
 \leq d_{AP}(\mu, \nu) \leq D_b(\mu, \nu)$ by Theorem~\ref{thm:metric}.
\item $\phi \equiv  \phi'\oplus p$. Assume $\phi'(\mu)\geq \phi'(\nu)$. Then $\phi(\mu)\geq \phi(\nu)$. By induction, we have $\phi'(\mu) - \phi'(\nu)\leq  D_b(\mu, \nu)$.
Thus
$$|\phi(\mu) - \phi(\nu)| = \min\{\phi'(\mu)+p, 1\} -  \min\{\phi'(\nu)+p, 1\}\leq  \phi'(\mu) - \phi'(\nu) \leq  D_b(\mu, \nu).$$

\item $\phi \equiv \neg \phi'$. By induction, we have $|\phi'(\mu) - \phi'(\nu)|\leq  D_b(\mu, \nu)$, thus
$|\phi(\mu) - \phi(\nu)|  = |1- \phi'(\mu) - 1 + \phi'(\nu)| \leq  D_b(\mu, \nu)$ as well.

\item $\phi \equiv   \bigwedge_{i\in I} \phi_i$. Assume $\phi(\mu)\geq  \phi(\nu)$. For any $\epsilon>0$, let $j\in I$ such that $\phi_j(\nu) \leq \phi(\nu)+\epsilon$. By induction, we have $|\phi_j(\mu) - \phi_j(\nu)|\leq  D_b(\mu, \nu)$. Then
$$|\phi(\mu) - \phi(\nu)|  \leq  \phi_{j}(\mu) - \phi_{j}(\nu) +\epsilon \leq  D_b(\mu, \nu) +\epsilon,$$
and $|\phi(\mu) - \phi(\nu)|  \leq D_b(\mu, \nu)$ from the arbitrariness of $\epsilon$.

\item $\phi \equiv \<a\>\phi'$. Assume $\phi(\mu)\geq \phi(\nu)$. 
Let $\mu'_*\in \dist(S)$ such that $\TRANA{\mu}{a}{\mu'_*}$ and $\gamma \cdot\phi'(\mu'_*) = \phi(\mu)$. From Theorem~\ref{thm:metric}, we have 
$\mu\sim_{D_b(\mu, \nu)}\nu$. Thus there exists $\nu'_*$ such that $\TRANA{\nu}{a}{\nu'_*}$ and $\mu'_*\sim_{D_b(\mu, \nu)/\gamma}\nu'_*$. Hence $\gamma \cdot D_b(\mu'_*, \nu'_*) \leq D_b(\mu, \nu)$, and
$$|\phi(\mu) - \phi(\nu)|  \leq  \gamma \cdot[\phi'(\mu'_*) - \phi'(\nu'_*)]\leq  \gamma \cdot D_b(\mu'_*, \nu'_*) \leq  D_b(\mu, \nu)$$
where the second inequality is from induction.
\end{itemize}

Now we turn to the proof of $D_b(\mu, \nu) \leq D_l(\mu, \nu)$. We
will achieve this by showing that the symmetric relations
$R_\epsilon=\{(\mu,\nu) \mid D_l(\mu, \nu) \leq \epsilon \},$ where $\epsilon\geq 0$, constitute an approximate bisimulation. Let $\mu R_\epsilon \nu$ for some $\epsilon \geq 0$. First, for any $\B\subseteq 2^{AP}$ we have 
$$\left|\sum_{A\in \B} \mu(A) - \sum_{A\in \B} \nu(A)\right| = |\B(\mu) - \B(\nu)| \leq D_l(\mu, \nu)\leq \epsilon.$$ 
Thus $d_{AP}(\mu, \nu) \leq \epsilon$ as well.
Now suppose $\TRANA{\mu}{a}{\mu'}$ for some $\mu'$. 
We have to show that there is some $\nu'$ with  $\TRANA{\nu}{a}{\nu'}$ and $D_l(\mu', \nu')\leq \epsilon/\gamma$. Consider the set
$$\k = \{\omega\in \dist(S) \mid \TRANA{\nu}{a}{\omega} \mbox{ and } D_l(\mu', \omega)> \epsilon/\gamma\}.$$
For each $\omega\in \k$, there must be some $\phi_{\omega}$ such that $|\phi_{\omega}(\mu') - \phi_{\omega}(\omega)|> \epsilon/\gamma$. 
As our logic includes the operator $\neg$, we can always assume that $\phi_{\omega}(\mu') > \phi_{\omega}(\omega)+\epsilon/\gamma$.
Let $p= \sup_{\omega\in \k} \phi_{\omega}(\mu')$. Let
$$\phi_{\omega}'=\phi_{\omega} \oplus [p - \phi_{\omega}(\mu')], \ \ \ \phi' = \bigwedge_{\omega\in \k} \phi'_{\omega},\ \ \  \mbox{  and  }\ \ \  \phi = \<a\>\phi'.$$
Then from the assumption that $D_l(\mu, \nu) \leq \epsilon$, we have $|\phi(\mu) - \phi(\nu)| \leq \epsilon$. Furthermore, we check that for any $\omega\in \k$,
$$\phi'_\omega(\mu')=\phi_{\omega}(\mu') \oplus [p - \phi_{\omega}(\mu')]=p.$$
Thus $\phi(\mu) \geq \gamma\cdot\phi'(\mu')=\gamma\cdot p$.

Let 
$\nu'$ be the distribution such that $\TRANA{\nu}{a}{\nu'}$ and $\phi(\nu) =\gamma\cdot \phi'(\nu')$. We are going to show that $\nu'\not\in \k$, and then $D_l(\mu', \nu')\leq \epsilon/\gamma$ as required. For this purpose, assume conversely that $\nu'\in \k$. 
Then
\begin{eqnarray*}
\phi(\nu) &=& \gamma\cdot\phi'(\nu')\leq \gamma\cdot\phi'_{\nu'}(\nu')\leq \gamma\cdot [\phi_{\nu'}(\nu')  +  p - \phi_{\nu'}(\mu')]\\
&<&\gamma\cdot p-\epsilon\leq \phi(\mu)-\epsilon,
\end{eqnarray*}
contradicting the fact that $|\phi(\mu) - \phi(\nu)| \leq \epsilon$.

We have proven that $\{R_\epsilon \mid \epsilon\geq 0\}$ is an approximate bisimulation. Thus $\mu\sim_\epsilon \nu$, and so $D_b(\mu, \nu) \leq \epsilon$, whenever $D_l(\mu, \nu) \leq \epsilon$. So we have $D_b(\mu, \nu) \leq D_l(\mu, \nu)$ from the arbitrariness of $\epsilon$. 
\qed
\end{proof}

\subsection{A Fixed Point-Based Approach}
In the following, we denote by $\mathcal{M}$ the set of pseudometrics
over $\dist(S)$.  
Denote by $\z$ the zero pseudometric which assigns 0 to each pair of distributions. 
For any $d, d'\in\mathcal{M}$, we write $d\leq d'$
if $d(\mu, \nu)\leq d'(\mu, \nu)$ for any $\mu$ and $\nu$. Obviously
$\leq$ is a partial order, and $(\mathcal{M}, \leq)$ is a complete
lattice.

\begin{definition}\label{def:metricfunc} Let $\A=(S, Act, \rightarrow,L,\alpha)$ be an input enabled probabilistic automaton.  
We define the function $F:\mathcal{M}\to
  \mathcal{M}$ as follows. For any $\mu, \nu\in \dist(S)$, 
  \begin{eqnarray*}
 F(d)(\mu,\nu) = \max_{a\in Act} &\{& d_{AP}(\mu, \nu),\\
&&  \sup_{\TRANA{\mu}{a}{\mu'}} \inf_{\TRANA{\nu}{a}{\nu'}} \gamma\cdot
  d(\mu',\nu'), \sup_{\TRANA{\nu}{a}{\nu'}}
  \inf_{\TRANA{\mu}{a}{\mu'}} \gamma\cdot d(\mu',\nu')\}.
\end{eqnarray*}
Then, $F$ is monotonic
with respect to $\leq$, and by Knaster-Tarski theorem, $F$ has a
least fixed point, denoted $D_f^\A$, given by 
\begin{align*}
D^\A_f = \bigvee_{n=0}^\infty
F^n(\z) \ .  
\end{align*}
\end{definition}

Once again, the fixed point-based distance for a general probabilistic automaton can be
defined in terms of the input enabled extension; that is, $D_f^\A(\mu, \nu) := D_f^{\A_\bot}(\mu, \nu)$. We always omit the superscripts for simplicity.

Similar to
Lemma~\ref{lem:continuous}, we can show that the supremum
(resp. infimum) in Definiton~\ref{def:metricfunc} can be replaced by
maximum (resp. minimum).  
Now  we show
that $D_f$ coincides with $D_b$.
\begin{theorem}\label{thm:dfdb}
$D_f= D_b$.  
\end{theorem}
As both $D_f$ and $D_b$ are defined in terms of the input enabled extension of automata, we only need to prove Theorem~\ref{thm:dfdb} for
input enabled case, which will be obtained by combining Lemma \ref{lem:left} and Lemma \ref{lem:right} below.
\begin{lemma}\label{lem:left}
For input enabled probabilistic automata, $D_f\leq D_b$.
\end{lemma}
\begin{proof}
It suffices to prove by induction that for any $n\geq 0$, $F^n(\z)\leq D_b$. The case of $n=0$ is trivial. Suppose
$F^n(\z)\leq D_b$ for some $n\geq 0$. Then for any $a\in Act$ and any $\mu, \nu$, we have
\begin{enumerate}
\item[(1)] $d_{AP}(\mu, \nu) \leq D_b(\mu, \nu)$ by the fact that $\mu\sim_{D_b(\mu, \nu)} \nu$;
\item[(2)] Note that $\mu\sim_{D_b(\mu, \nu)} \nu$. Whenever $\TRANA{\mu}{a}{\mu'}$, we have $\TRANA{\nu}{a}{\nu'}$ for some $\nu'$ such that $\mu'\sim_{D_b(\mu, \nu)/\gamma} \nu'$, and hence $\gamma\cdot D_b(\mu', \nu')\leq D_b(\mu, \nu)$. That, together with the assumption $F^n(\z)\leq D_b$, implies
$$ \max_{\TRANA{\mu}{a}{\mu'}} \min_{\TRANA{\nu}{a}{\nu'}}\gamma\cdot  F^n(\z)(\mu',\nu') \leq D_b(\mu,\nu).$$
The symmetric form can be similarly proved.
\end{enumerate}
Summing up (1) and (2), we have $F^{n+1}(\z)\leq D_b$. \qed
\end{proof}

The opposite direction is summarized in the following lemma. The proof is technically involved so we put it into the appendix.

\begin{lemma}\label{lem:right}
For input enabled probabilistic automata, $D_b\leq D_f$.
\end{lemma}

\subsection{Comparison with State-Based Metrics}\label{app:comparison}
In this section, we prove that our distribution-based bisimulation
metric is lower bounded by the state-based game bisimulation
metrics~\cite{AlfaroMRS07} for MDPs. This game bisimulation metric is
particularly attractive as it preserves probabilistic reachability,
long-run, and discounted average behaviors \cite{ChatterjeeAMR10}. We first recall the
definition of state-based game bisimulation metrics \cite{AlfaroMRS07}
for MDPs:

\begin{definition}
Given $\mu, \nu\in \dist(S)$, $\mu\otimes \nu$ is defined as the set of \emph{weight} functions $\lambda: S\times S\rightarrow [0,1]$ such that
for any $s,t\in S$, 
$$\sum_{s\in S} \lambda(s,t) = \nu(t)\ \ \ \mbox{ and  }\ \ \ \sum_{t\in S} \lambda(s,t) = \mu(s).$$
Given a metric $d$ defined on $S$, we lift it to $\dist(S)$ by defining
$$d(\mu, \nu) = \inf_{\lambda\in \mu\otimes \nu} \left(\sum_{s,t\in S} \lambda(s,t)\cdot d(s, t)\right).$$
\end{definition}

Actually the infimum in the above definition is attainable. 

\begin{definition} 
We define the function $f:\mathcal{M}\to
  \mathcal{M}$ as follows. For any $s, t \in S$, 
  \begin{eqnarray*}
 f(d)(s,t) = \max_{a\in Act}\left\{1-\delta_{L(s), L(t)}, \sup_{\TRANPA{s}{a}{\mu}} \inf_{\TRANPA{t}{a}{\nu}} 
  \gamma\cdot d(\mu,\nu), \sup_{\TRANPA{t}{a}{\nu}}
  \inf_{\TRANPA{s}{a}{\mu}}  \gamma\cdot d(\mu,\nu)\right\}
\end{eqnarray*}
where $\delta_{L(s), L(t)}=1$ if $L(s)=L(t)$, and 0 otherwise. We take $\inf \emptyset = 1$ and $\sup\emptyset = 0$.
Again, $f$ is monotonic
with respect to $\leq$, and by Knaster-Tarski theorem, $F$ has a
least fixed point, denoted $d_f$, given by 
\begin{align*}
d_f = \bigvee_{n=0}^\infty
f^n(\z) \ .  
\end{align*}

\end{definition}

Now we can prove the quantitative extension of Lemma~\ref{lem:bsp}. Without loss of any generality, we assume that $\A$ itself is input enabled. Let $d_n = f^n(\z)$ and $D_n = F^n(\z)$ in Definition~\ref{def:metricfunc}.


\begin{lemma}\label{lem:tmp11}
For any $n\geq 1$, $d_{AP}(\mu, \nu )\leq d_n(\mu, \nu)$.
\end{lemma}
\begin{proof}
Let $\lambda$ be the weight function such that
$d_n(\mu, \nu) = \sum_{s,t\in S} \lambda(s,t)\cdot d_n(s, t)$. 
Since $d_n(s, t)\geq 1-\delta_{L(s), L(t)}$, we have
$$d_n(\mu, \nu)\geq 1- \sum_{s,t:L(s)= L(t)} \lambda(s,t).$$
On the other hand, for any $A\subseteq AP$, recall that $S(A)=\{s\in S \mid L(s) = A\}$. Then
\begin{eqnarray*}
\mu(A) - \nu(A) &=& \sum_{s\in S(A)} \mu(s) - \sum_{t\in S(A)} \nu(t) \\
&=&  \sum_{s\in S(A)}\sum_{t\not\in S(A)}\lambda(s,t) - \sum_{t\in S(A)}\sum_{s\not\in S(A)}\lambda(s,t).
\end{eqnarray*}
Let $\B\subseteq 2^{AP}$ such that
$d_{AP}(\mu, \nu) = \sum_{A\in \B} [\mu(A) - \nu(A)]$. Then
$$d_{AP}(\mu, \nu) \leq \sum_{A\in \B} \sum_{s\in S(A)}\sum_{t\not\in S(A)}\lambda(s,t) \leq \sum_{s,t:L(s)\neq L(t)} \lambda(s,t),$$
and the result follows. \qed
\end{proof}

\begin{theorem}\label{thm:less}
 Let $\A$ be a probabilistic automaton. Then $D_f \leq d_f$.
\end{theorem}
\begin{proof}
We prove by induction on $n$ that  $D_n(\mu, \nu) \leq d_n(\mu, \nu)$ for any $\mu, \nu\in \dist(S)$ and $n\geq 0$. The case $n=0$ is obvious. Suppose the result holds for some $n-1\geq 0$. 
Then from Lemma~\ref{lem:tmp11}, we need only to show that for any $\TRANA{\mu}{a}{\mu'}$ there exists $\TRANA{\nu}{a}{\nu'}$
such that $\gamma\cdot D_{n-1}(\mu', \nu') \leq d_n(\mu, \nu)$.

Let $\TRANA{\mu}{a}{\mu'}$. Then for any $s\in S$, $\TRANPA{s}{a}{\mu_s}$ with $\mu' = \sum_{s\in S} \mu (s)\cdot \mu_s$.
By definition of $d_n$, for any $t\in S$, we have $\TRANPA{t}{a}{\nu_t}$ such that $\gamma\cdot d_{n-1}(\mu_s, \nu_t) \leq d_n(s, t)$. Thus
$\TRANA{\nu}{a}{\nu'} := \sum_{t\in S} \nu(t)\cdot \nu_t$, and by induction, $D_{n-1}(\mu', \nu')\leq d_{n-1}(\mu', \nu')$. Now it suffices to prove  $\gamma\cdot d_{n-1}(\mu', \nu')\leq d_n(\mu, \nu)$.

Let $\lambda\in \mu\otimes \nu$ and $\gamma_{s,t}\in \mu_s\otimes \nu_t$ be the weight functions such that
$$d_n(\mu, \nu)=\sum_{s,t\in S} \lambda(s,t)\cdot d_n(s, t),\ \  d_{n-1}(\mu_s, \nu_t)=\sum_{u, v\in S} \gamma_{s,t}(u,v)\cdot d_{n-1}(u, v).$$
Then 
\begin{eqnarray*}
d_n(\mu, \nu)&\geq &\gamma\cdot  \sum_{s,t\in S} \lambda(s,t)\cdot d_{n-1}(\mu_s, \nu_t)\\
&=&\gamma\cdot \sum_{u, v\in S}\sum_{s,t\in S} \lambda(s,t)\gamma_{s,t}(u,v)\cdot d_{n-1}(u, v).
\end{eqnarray*}
We need to show that the function $\eta(u, v) :=\sum_{s,t\in S} \lambda(s,t)\gamma_{s,t}(u,v)$ is a weight function for $\mu'$ and $\nu'$. It is easy to check that
\begin{eqnarray*}
\sum_u \eta(u, v) &=& \sum_{s,t\in S} \lambda(s,t)\sum_u\gamma_{s,t}(u,v)= \sum_{s,t\in S} \lambda(s,t)\nu_t(v) \\
&= & \sum_{t\in S} \nu(t)\nu_t(v) = \nu'(v).
\end{eqnarray*}
Similarly, we have $\sum_v \eta(u, v) = \mu'(u)$.
\qed
\end{proof}

\begin{example}
  Consider Fig.~\ref{fig:exam1}, and assume $\epsilon_1\ge
\epsilon_2>0$. Applying the definition of $D_b$, it is easy to check
that $D_b(\dirac{q},\dirac{q'})=0.5(\epsilon_1-\epsilon_2)\gamma$. By our
results, we have $D_l(\dirac{q},\dirac{q'}) =
D_f(\dirac{q},\dirac{q'}) = D_b(\dirac{q},\dirac{q'})$. Note that for
the discounting case $\gamma<1$, difference far in the future will
have less influence in the distance.

We further compute the distance under state-based bisimulation metrics
(see \cite{FernsPP11} for example). Assume that $\gamma=1$. One first
compute the distance between $r_1$ and $r'$ being
$\frac16+\epsilon_1$, between $r_2$ and $r'$ being
$\frac16+\epsilon_2$. Then, the state-based bisimulation metric
between $q$ and $q'$ is $\frac16+0.5(\epsilon_1+\epsilon_2)$, which
can be obtained by lifting the state-based metrics. 
\end{example}

\subsection{Comparison with Equivalence Metric}

Note that we can easily extend the equivalence relation defined in Definition~\ref{def:distribution_based}
to a notion of equivalence metric:
\begin{definition}[Equivalence Metric]
  Let $\A_i=(S_i,Act_i,\rightarrow_i,L_i,\alpha_i)$ with $i=1,2$ be
  two reactive automata with $Act_1=Act_2=:Act$, and $F_i=\{s\in
  S_i\mid L(s)=AP\}$ the set of final states for $\A_i$. We say
  $\A_1$ and $\A_2$ are $\epsilon$-equivalent, denoted $\A_1
  \sim_\epsilon^d \A_2$, if 
  for any input word $w=a_1a_2\ldots a_n$, 
  $|\A_1(w)-\A_2(w)| \le \epsilon$. Furthermore, the equivalence distance
  between $\A_1$ and $\A_2$ is defined by $D_d(\A_1, \A_2):=\inf\{\epsilon\geq 0\mid \A_1
  \sim_\epsilon^d \A_2\}$.
\end{definition}

Now we show that  for
reactive automata, the equivalence metric 
$D_d$ coincide with our undiscounted bisimulation metric $D_b$, which  
may be regarded as a quantitative extension of Lemma~\ref{lem:eqra}.
\begin{proposition}
Let $\A_1$ and $\A_2$ be two reactive automata with the same set of
  actions $Act$. Let the discount factor $\gamma =1$. Then
  $D_d(\A_1, \A_2) = D_b(\alpha_1, \alpha_2)$ where $D_b$ is defined in the direct sum of $\A_1$ and $\A_2$.
\end{proposition}
\begin{proof}
We first show that $D_d(\A_1, \A_2) \leq D_b(\alpha_1, \alpha_2)$.
For each input word $w=a_1a_2\ldots a_n$, it is easy to check that $ \A_i(w)=\phi(\alpha_i)$
where $\phi= \langle a_1\rangle \langle
a_2\rangle\ldots \langle a_n\rangle (F_1\cup F_2)$.
As we have shown that $D_b=D_l$, it holds $|\A_1(w)-\A_2(w)| \le D_b(\alpha_1, \alpha_2)$, and hence $\A_1
  \sim_{D_b(\alpha_1, \alpha_2)}^d \A_2$. Then $D_d(\A_1, \A_2) \leq D_b(\alpha_1, \alpha_2)$ by definition. 
  
Now we turn to the proof of $D_d(\A_1, \A_2) \geq D_b(\alpha_1, \alpha_2)$. First we show that
$$R_\epsilon=\{ (\mu, \nu) \mid \mu\in \dist(S_1), \nu\in \dist(S_2), \A_1^\mu \sim_\epsilon^d \A_2^\nu\}$$
is an approximate bisimulation. Here for a probabilistic automaton $\A$, we denote by $\A^\mu$ the automaton which
is the same as $\A$ except that the initial distribution is replaced by $\mu$. Let $\mu R_\epsilon\nu$.
Since $L(s) \in \{\emptyset, AP\}$ for all $s\in S_1\cup S_2$, we have $\mu(AP) + \mu(\emptyset) = \nu(AP) + \nu(\emptyset) = 1$.
Thus
$$d_{AP}(\mu, \nu) =  |\mu(AP) - \nu(AP)|=  |\mu(F_1) - \nu(F_2)|.$$ 
Note that $\mu(F_1)=\A_1^\mu(e)$  and $\nu(F_2)=\A_2^\nu(e)$,  where $e$
is the empty string. Then $d_{AP}(\mu, \nu)=|\A_1^\mu(e)-\A_2^\nu(e)|\leq \epsilon$.

Let $\TRANA{\mu}{a}{\mu'}$ and $\TRANA{\nu}{a}{\nu'}$. We need to show $\mu' R_\epsilon\nu'$, that is, 
$\A_1^{\mu'} \sim_\epsilon^d \A_2^{\nu'}$. For any $w\in Act^*$, note that $\A_1^{\mu'}(w) = \A_1^{\mu}(aw)$. Then 
 $$|\A_1^{\mu'}(w)-\A_2^{\nu'}(w)| = |\A_1^{\mu}(aw)-\A_2^{\nu}(aw)|  \le \epsilon,$$
 and hence  $\A_1^{\mu'} \sim_\epsilon^d \A_2^{\nu'}$ as required. 
 
 Having proven that $R_\epsilon$ is an approximate bisimulation, we know $\A_1 \sim_\epsilon^d \A_2$ implies $\alpha_1\sim_\epsilon \alpha_2$.
 Thus $D_d(\A_1, \A_2) =\inf\{\epsilon \mid \A_1 \sim_\epsilon^d \A_2\} \geq \inf\{\alpha_1\sim_\epsilon \alpha_2\}=D_b(\alpha_1, \alpha_2)$.\qed
 \end{proof}
\section{Discussion and Future Work}\label{sec:conclusion}
In this paper, we considered Segala's automata, and proposed a novel
notion of bisimulation by joining the existing notions of equivalence
and bisimilarities. We have demonstrated the utility of our definition
by studying distribution-based bisimulation metrics, which have been
extensively studied for MDPs. 

As future work we would like to identify further solutions and techniques
developed in one area that could  inspire solutions for the corresponding
problems in the other area. This includes for instance decision
algorithm developed for equivalence
checking~\cite{Tzeng92,KieferMOWW11}, extensions to simulations, and
compositional verification for probabilistic automata.

\bibliographystyle{abbrv}
\bibliography{bib}

\appendix
\section{Omitted Proofs}

%

\renewcommand{\proofname}{Proof}

\subsection{Proof of Lemma \ref{lem:right}}

To prove Lemma \ref{lem:right}, we first introduce the notion of bounded approximation bisimulations.
\begin{definition} Let $\A$ be an input enabled probabilistic automaton.
We define symmetric relations
\begin{itemize}
\item $\ysim{\epsilon}_0 {:=} \dist(S)\times \dist(S)$ for any $\epsilon \geq 0;$
\item for $n\geq 0$, $\mu\ysim{\epsilon}_{n+1} \nu$ if $1 \leq \epsilon$, and whenever $\TRANA{\mu}{a}{\mu'}$, there exists $\TRANA{\nu}{a}{\nu'}$ for some $\nu'$ such that $\mu'\ysim{\epsilon/\gamma}_{n} \nu'$.
\item $\ysim{\epsilon} {:=} \bigcap_{n\geq 0} \ysim{\epsilon}_n$.
\end{itemize}
\end{definition}

The following lemma collects some useful properties of $\ysim{\epsilon}_n$ and $\ysim{\epsilon}$.

\begin{lemma}\label{lem:simlimit}
\begin{enumerate}
\item $\ysim{\epsilon}_n {\subseteq} \ysim{\epsilon}_m$ provided that $n \geq m$;
\item for any $n\geq 0$, $\ysim{\epsilon}_n {\subseteq} \ysim{\epsilon'}_n$ provided that $\epsilon \leq \epsilon'$;
\item for any $n\geq 0$, $\ysim{\epsilon}_n$ is continuous;
\item $\ysim{\epsilon} {=} \sim_\epsilon$.
\end{enumerate}
\end{lemma}
\begin{proof} Items 1, 2, and 3 are easy by induction, and so is the $ \sim_\epsilon {\subseteq} \ysim{\epsilon}$ part of Item 4. To prove $\ysim{\epsilon} {\subseteq} \sim_\epsilon$, we show that $\{\ysim{\epsilon} \mid \epsilon\geq 0\}$ is an  approximate bisimulation. Suppose $\mu\ysim{\epsilon}\nu$. Then $d_{AP}(\mu, \nu) \leq \epsilon$ by definition. Now let $\TRANA{\mu}{a}{\mu'}$.
 For each $n\geq 0$, from the assumption that $\mu\ysim{\epsilon}_{n+1}\nu$ we have $\TRANA{\nu}{a}{\nu_n}$ such that  $\mu'\ysim{\epsilon/\gamma}_{n}\nu_n$. Let $\{\nu_{i_k}\}_k$ be a convergent subsequence of $\{\nu_n\}_n$ such that $\lim_k \nu_{i_k} = \nu'$ for some $\nu'$. Then from the continuity of $\TRANA{}{a}{}$ we have $\TRANA{\nu}{a}{\nu'}$. We claim further that $\mu'\ysim{\epsilon/\gamma} \nu'$. Otherwise there exists $N$ such that 
$\mu'\not\ysim{\epsilon/\gamma}_{N}\nu'$. Now by the continuity of $\ysim{\epsilon/\gamma}_{N}$, we have $\mu'\not\ysim{\epsilon/\gamma}_{N}\nu_j$ for some $j\geq N$. This contradicts the fact that  $\mu'\ysim{\epsilon/\gamma}_j\nu_{j}$ and item 1. \qed
\end{proof}

\begin{lemma}\label{lem:tmp}
For any $n\geq 0$, we have $\mu \ysim{F^n(\z)(\mu, \nu)}_n \nu$.
\end{lemma}
\begin{proof} 
We prove this lemma by induction on $n$. The case of $n=0$ is trivial. Suppose $\mu \ysim{F^n(\z)(\mu, \nu)}_n \nu$ for some $n\geq 0$. Let $a\in Act$. By definition, we have  
$$F^{n+1}(\z)(\mu,\nu) \geq \max_{\TRANA{\mu}{a}{\mu'}} \min_{\TRANA{\nu}{a}{\nu'}} \gamma\cdot F^n(\z)(\mu',\nu').
$$
Thus for any $\TRANA{\mu}{a}{\mu'}$, there exists $\TRANA{\nu}{a}{\nu'}$ such that $\gamma\cdot F^n(\z)(\mu',\nu')\leq F^{n+1}(\z)(\mu,\nu)$. By induction, we know $\mu' \ysim{F^n(\z)(\mu', \nu')}_n \nu'$, thus $\mu' \ysim{ F^{n+1}(\z)(\mu,\nu)/\gamma}_n \nu'$ from Lemma~\ref{lem:simlimit}(2). On the other hand, we have $F^{n+1}(\z)(\mu,\nu) \geq d_{AP}(\mu, \nu)$ by definition.
Thus we have $\mu \ysim{ F^{n+1}(\z)(\mu,\nu)}_{n+1} \nu$.
\qed
\end{proof}

With the two lemmas above, Lemma~\ref{lem:right} follows easily.
\renewcommand{\proofname}{Proof of Lemma~\ref{lem:right}}
\begin{proof}
For any $\mu$ and $\nu$, by Lemmas~\ref{lem:tmp} and \ref{lem:simlimit}(2), we have $\mu \ysim{D_f(\mu, \nu)}_n \nu$ for all $n\geq 0$, so $\mu \ysim{D_f(\mu, \nu)} \nu$ by definition. Then from 
 Lemma~\ref{lem:simlimit}(4) we have
$\mu \sim_{D_f(\mu, \nu)} \nu$, hence $D_b(\mu, \nu)\leq D_f(\mu, \nu)$. \qed
\end{proof}

\end{document}